\documentclass[journal]{IEEEtran}
\ifCLASSINFOpdf
\else
\fi

\usepackage{amsmath,amssymb,amsthm}
\usepackage[utf8]{inputenc}
\usepackage[T1]{fontenc}
\usepackage{amsfonts}
\usepackage{cite}
\usepackage{booktabs}
\usepackage[ruled,vlined]{algorithm2e}
\usepackage{algpseudocode}  
\usepackage{graphicx}
\usepackage{booktabs}
\usepackage{multirow}
\usepackage{subfigure}
\usepackage{booktabs}
\usepackage{comment}

\newcommand{\RNum}[1]{\uppercase\expandafter{\romannumeral #1\relax}}

\newtheorem{theorem}{Theorem}
\newtheorem{lemma}{ Lemma}

\newtheorem{corollary}{Corollary}
\newtheorem{definition}{Definition}

\begin{document}

\title{Gain without Pain: Offsetting DP-injected Nosies Stealthily in Cross-device Federated Learning}

\author{Wenzhuo Yang, Yipeng Zhou, Miao Hu, Di Wu, James Xi Zheng, Hui Wang and Song Guo 
\thanks{Wenzhuo Yang, Miao Hu and Di Wu are with the School of Computer Science and Engineering, Sun Yat-sen University, Guangzhou, China (e-mail: yangwzh8@mail2.sysu.edu.cn, humiao5@mail.sysu.edu.cn, wudi27@mail.sysu.edu.cn).}
\thanks{
Yipeng Zhou and James Xi Zheng are with the Department of Computing, FSE, Macquarie University, Sydney, Australia (e-mail: \{yipeng.zhou, xi.zheng\}@mq.edu.au).}
\thanks{Hui Wang is with the Institute for Network Sciences and Cyberspace, Tsinghua University, Beijing 100084, China, and also with the Beijing National Research Center for Information Science and Technology, Beijing 100084, China (e-mail: jessiewang@tsinghua.edu.cn).}
\thanks{Song Guo is with the Department of Computing, The Hong Kong Polytechnic University, Hong Kong (email: cssongguo@comp.polyu.edu.hk).}

}

\maketitle

\thispagestyle{empty}
\pagestyle{empty}

\begin{abstract}
Federated Learning (FL) is an emerging paradigm through which decentralized devices can collaboratively train a common model. However, a serious concern  is the leakage of privacy from exchanged gradient information between clients and the parameter server (PS) in FL. To protect gradient information, clients can adopt differential privacy (DP) to add additional noises and distort original gradients before they are uploaded to the PS. Nevertheless, the model accuracy will be significantly impaired by DP noises, making DP impracticable in real systems. 
In this work, we propose a novel \underline{N}oise \underline{I}nformation \underline{S}ecretly \underline{S}haring (NISS) algorithm to alleviate the disturbance  of DP noises by sharing negated noises among clients.  We theoretically prove that: 1) If clients are trustworthy, DP noises can be perfectly offset on the PS; 2) Clients can easily distort negated DP noises to protect themselves in case that other clients are not totally trustworthy, though the cost lowers model accuracy. NISS is particularly applicable for FL across multiple IoT (Internet of Things) systems, in which all IoT devices need to collaboratively train a model. To verify the effectiveness and the superiority of the NISS algorithm, we conduct experiments with the MNIST and CIFAR-10 datasets. The experiment results verify our analysis and demonstrate that NISS can improve model accuracy by 21\% on average and obtain better privacy protection if clients are trustworthy.  

\end{abstract}
\IEEEpeerreviewmaketitle


\section{Introduction}
\IEEEPARstart{W}{ith} the remarkable development of IoT (Internet of Things) systems, IoT devices such as mobile phones, cameras and IIoT (Industrial IoT) devices have been widely deployed in our daily life \cite{atzori2010internet}\cite{gubbi2013internet}. On one hand, IoT devices with powerful computing and communication capacity are generating more and more data. On the other hand, to provide more intelligent services, decentralized IoT devices have motivation to collaborate via federated learning (FL) so that distributed data can be fully exploited for model training \cite{zhang2019deep}\cite{li2018contract}. 

The training process via FL can be briefly described as follows. In a typical FL system, a parameter server (PS) is deployed to aggregate gradients uploaded by clients, and distribute aggregated results back to clients \cite{kairouz2019advances, mcmahan2017communication, konevcny2016federated1, konevcny2016federated2}. The model training process terminates after exchanging the gradient information between clients and the PS for a certain number of rounds. However, it has been studied in \cite{melis2019exploiting, hitaj2017deep, zhu2019deep, zhao2020idlg} that it can lead to the leakage of user privacy if the gradient information is disclosed. In addition, the PS is not always trustworthy \cite{li2020federated, kairouz2019advances}, which  also possibly invades user privacy.

Recently, it has been extensively investigated by academia and industry to adopt differential privacy  (DP) on each client \cite{wei2020federated, 10.1561/0400000042, abadi2016deep, wu2020value, geyer2017differentially} so as to protect the gradient information. DP can distort original gradients by adding additional noises, which however unavoidably distorts the aggregated gradients on the PS and hence impairs the model accuracy \cite{wei2020federated}. It has been reported in the work \cite{wei2020federated, zhu2019deep} that DP noises can significantly lower model accuracy by $10\sim 30\%$. It implies that straightly implementing DP in real systems is impracticable when high model accuracy is required \cite{abadi2016deep}.

\begin{figure}[ht]
\centering  
\includegraphics[width=0.5\textwidth]{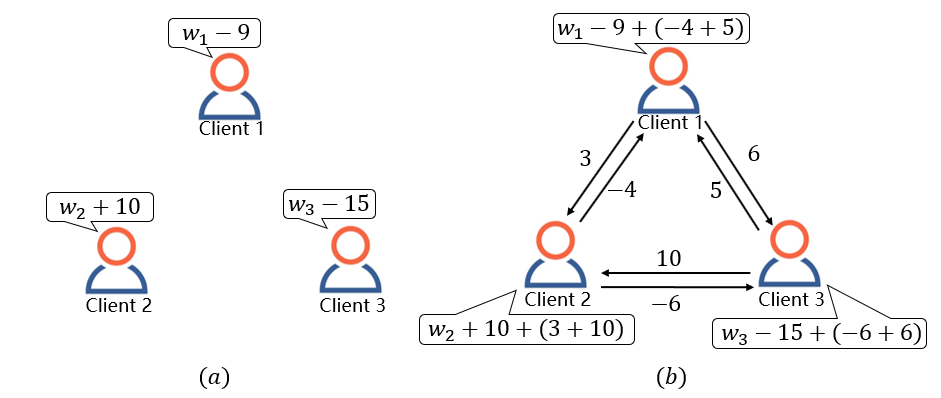}
\caption{A case with clients whose noises can be perfectly offset among themselves. Here $w_k$ represents the parameter for client $k$.}
\label{Fig:ideal}
\end{figure}

To alleviate the disturbance of DP noises on the aggregated gradients without compromising user privacy, we propose an algorithm to secretly offset DP noises.  
The idea of our work can be explained by the example shown in Fig.~\ref{Fig:ideal}. There are three clients and the model to be trained by these clients is represented by the parameter $w_k$ for client $k$. Each client distorts the original gradient information by adding a random number, as presented in Fig.~\ref{Fig:ideal}(a). We suppose that the random number is generated according to the Gaussian distribution determined by the client's privacy budget. However, the noise can be offset if a client can negate and split its noise into multiple shares, and distribute these negated shares with other clients, as presented in Fig.~\ref{Fig:ideal}(b). Each client uploads its gradients plus all noises (\emph{i.e.}, its own noises and negated noise shares from other clients) to the PS, and then these noises can be perfectly offset among themselves.

Inspired by this example,  we propose the \underline{N}oise \underline{I}nformation \underline{S}ecretly \underline{S}haring (NISS) algorithm through which clients can secretly share their noise information with each other. We theoretically prove that:  1) If clients are trustworthy, DP noises can be perfectly offset on the PS without compromising privacy protection; 2) Clients can easily distort  negated noise shares received from other clients in case that other clients are not totally trustworthy. We also investigate the extreme case that the PS colludes with other clients to crack the gradient information of a particular client. In this extreme case, there is a trade-off between model accuracy and privacy protection, and model accuracy cannot be improved without compromising privacy protection. However, we would like to emphasize that NISS is particularly applicable for FL across multiple IoT systems. IoT devices within the same system can trust each other to certain extent so that the model accuracy can be improved accordingly. Besides, devices within the same system can be connected with high speed networks so that the communication overhead caused by NISS is acceptable.

Our main contributions are summarized as below:
\par
\begin{itemize}
    \item We propose the NISS algorithm that can secretly offset noises generated by DP adopted by each client so that the disturbance on the aggregated gradients can be removed. 
    \item We theoretically prove that the DP noises can be perfectly offset if clients are trustworthy. Even if clients are not totally trustworthy, clients can still protect themselves by distorting the negated noise shares transmitted between clients. 
    \item At last, we conduct experiments with the MNIST and CIFAR-10 datasets, and the experiment results demonstrate that NISS algorithm can obtain better privacy protection and higher accuracy.
\end{itemize}
\par
The reminder of this paper is organized as follows. In Section \ref{Sec:Related}, we introduce relate work on FL, DP, and SMC. In Section \ref{Sec:Preliminaries}, we introduce the preliminary knowledge. In Section \ref{Sec:NISS}, we elaborate the NISS algorithm. In Section \ref{Sec:Analysis}, we present the analysis of noise offsetting and security. In Section \ref{Sec:Experiment}, we show the simulations, compare our scheme with other schemes and discuss the experimental results. Finally, we conclude the paper in Section \ref{Sec:Conclusion}.

\section{Related work}
\label{Sec:Related}
\subsection{Federated Learning (FL)}
FL, as a recent advance of distributed machine learning, empowers participants to collaboratively train a model under the orchestration of a central parameter server, while keeping the training data decentralized\cite{kairouz2019advances}. It was first proposed by Google in 2016 \cite{mcmahan2017communication}. During the training process, each participant's raw data is stored locally and will not be exchanged or transferred for training. FL has the advantages of making full use of IoT computing power with preserved user privacy. 
\par
The work in \cite{mcmahan2017communication} firstly proposed FedAVG, which is one of the most widely used model average algorithms in FL. The work in \cite{li2019convergence} analyzed the convergence rate of FedAVG with non-IID data simple distributions. The work \cite{kairouz2019advances} and \cite{li2020federated} showed a comprehensive introduction to the history, technical methods and unresolved problems in FL. The work in \cite{geyer2017differentially} proved that the bare FedAVG can protect the privacy of participants to some extent. However only exchanging gradients information still has a high risk of privacy leakage \cite{melis2019exploiting, hitaj2017deep, zhu2019deep, zhao2020idlg}.
Despite tremendous efforts contributed by prior works, there exist many issues in FL that have not been solved very well, such as inefficient communication and device variability \cite{li2020federated, yang2019federated, wei2020federated}.
\par

\subsection{Differential Privacy (DP)}
DP is a very effective mechanism for privacy preservation that can be applied in FL \cite{wei2020federated, abadi2016deep, geyer2017differentially, wu2020value}. It uses a mechanism to generate random noises that are added to query results so as to distort original values.  
\par
The most commonly used mechanism for adding noises to FL is the Gaussian mechanism. The work in \cite{abadi2016deep} investigated how to apply Gaussian mechanism in machine learning systems. Then the work in \cite{geyer2017differentially} studied how to use the Gaussian mechanism in FL. In \cite{wei2020federated}, a FedSGD with DP algorithm is proposed for FL systems and its convergence rate is analyzed.
The work in \cite{zhu2019deep} introduced a novel method named DLG to measure the level of privacy preservation in FL. In FL with DP, a higher $\epsilon$ implies a smaller variance of DP noises, and hence a lower level of privacy preservation.  Model accuracy can be largely affected by DP noises \cite{wei2020federated, zhu2019deep}. 
\par
In the field of IoT, FL with DP has also attracted a lot of attention recently. In \cite{briggs2020review}, the author surveys a wide variety of papers on privacy preserving methods that are crucial for FL in IoT. The work in \cite{zhao2020privacy} designed a FL system with DP leveraging the reputation mechanism to assist home appliance manufacturers to train a machine learning model based on customers’ data. The work in \cite{zhao2020local} proposed  to integrate FL and DP to facilitate crowdsourcing applications to generate a machine learning model.

Basically, there is a trade-off between the extent of privacy protection and model accuracy if DP is straightly incorporated into FL. Different from these works, we devise a novel algorithm through which clients can generate negatively correlated DP noises to get rid of the negative influence on model accuracy. 


\subsection{Secure Multi-party Computing (SMC)}
Other than DP, SMC is another effective way for privacy preservation in FL.  In previous studies, SMC has been used in many machine learning models \cite{agrawal2000privacy, du2004privacy, vaidya2002privacy, vaidya2008privacy}. At present, \textit{Secret Sharing (SS)} and \textit{Homomorphic Encryption (HE)} are two main ways in SMC to protect privacy in FL.

HE performs complicated computation operations on gradients. During the gradient aggregation and transmission, it is always calculated in an independent encryption space, instead of directly using the raw gradients value \cite{hardy2017private, aono2017privacy}. SS is a method to generate several shares for a secret and send them to several participants. As long as most of participants are present, the secret can be recovered. In FL, participants can add masks to their gradients and share their masks as a secret to others. If the PS can receive returns from a sufficient number of participants, the masks can be eliminated. Several works based on SS in FL have been proposed in \cite{bonawitz2017practical, mandal2018nike, xu2019verifynet}.

However, SS and HE consume too much computing resources, which prohibit their deployment in real world \cite{vepakomma2018no}. In fact, our work is a combination of SS and DP, but the computation overhead of our noise sharing scheme is very low. 
\par

\par

\section{Preliminaries}
\label{Sec:Preliminaries}

To facilitate the understanding of our algorithm, the list of main notations used in our work is presented in {Table \ref{symbols}}.
\begin{table}[htbp]
\centering
\caption{LIST OF SYMBOLS}\label{tab:aStrangeTable}
\begin{tabular}{cc}
\toprule  
Symbol & Meaning  \\
\midrule  
$K$ & The number of clients\\
$k$ & The index of clients\\
$t$ & The index of global training round\\
$E$ & The number of local training round\\
$\eta$ & The learning rate\\
$h$ & The dimension of the parameters\\
$\ell$ & The loss function\\
$\nabla \ell$ & The gradient of function $\ell$\\
$m$ & The number of clients in each global round\\
$d_k$ & The cardinality of $\mathcal{D}_k$\\
$p_k$ & The aggregation weight of client $k$\\
$\sigma^2$ & The unit noise variance\\
$\sigma_k^2$ & The Gaussian noise variance of client $k$\\

$\mathcal{N}$ & Gaussian Distribution\\
$\mathcal{D}_k$ & The dataset of client $k$\\
$\mathcal{M}_t$ & The client set of $m$ client in round t\\
$\mathbf{w}$ & The global model parameters\\
$\mathbf{w}^k$ & The local model parameters of client $k$\\
$\mathbf{n}$ & The noise generated by DP mechanism\\
$\mathbf{r}$ & The negated noise\\
$s$ & A random variables to distort $\mathbf{r}$\\
$\tau_k^2$ & The variance of $s$\\
$\mathbb{I}_h$ & The $h \times h$ identity matrix\\
$\epsilon,\delta$ & DP parameters\\
\bottomrule 
\end{tabular}
\label{symbols}
\end{table}

\subsection{Differential Privacy}
\label{Sec:Preliminaries.dp}
It was assumed that user privacy will not be leaked if only gradient information is disclosed. However, it was investigated in \cite{melis2019exploiting,hitaj2017deep,zhu2019deep,zhao2020idlg} that privacy information can be reconstructed through gradient information. Therefore, it was proposed in \cite{abadi2016deep} that clients can adopt DP to further disturb their gradient information by adding additional noises to their disclosed information. According to the prior work \cite{10.1561/0400000042}, an algorithm satisfying $(\epsilon, \delta)$-differential privacy is defined as follows.  
\begin{definition}
\label{dpdefinition}
{ A randomized mechanism $\mathcal{M}: \mathcal{X} \rightarrow \mathcal{R}$ with domain $\mathcal{X}$ and range $\mathcal{R}$ satisfies $(\epsilon,\delta)$-differentially privacy if for any two adjacent databases $\mathcal{D}_i, \mathcal{D}_i^{\prime} \in \mathcal{X}$ and for any subset of outputs $S \subseteq \mathcal{R}$, }
\begin{equation}
    \operatorname{Pr}[\mathcal{M}(\mathcal{D}_i) \in S] \leq e^{\varepsilon} \operatorname{Pr}\left[\mathcal{M}\left(\mathcal{D}_i^{\prime}\right) \in S\right]+\delta.
\end{equation}
\end{definition}
Here, $\epsilon$ is the privacy budget which is the distinguishable bound of all outputs on adjacent databases $\mathcal{D}_i$ and $\mathcal{D}_i^{\prime}$. $\delta$ represents the probabilities that two adjacent outputs of the databases $\mathcal{D}_i, \mathcal{D}_i^{\prime}$ cannot be bounded by $\epsilon$ after using Algorithm $\mathcal{M}$. 
$\epsilon$ is also called the privacy budget. 
Intuitively, a DP mechanism $\mathcal{M}$ with a smaller privacy budget $\epsilon$ has a stronger privacy protection and vice verse.  

\begin{theorem}
\label{The:GaussM}
(Gaussian Mechanism). Let $\epsilon \in (0,1)$ be arbitrary and $\mathcal{D}_i$ denote the database. For $c^2 > 2 ln(1.25/\delta)$, the Gaussian Mechanism $\mathcal{M} = f(\mathcal{D}_i)+ \mathcal{N}(0, \sigma^2)$ with parameter $\sigma \geq \frac{c \Delta f}{  \epsilon}$ is $(\epsilon,\delta)$-differentially private. Here, $f(\mathcal{D}_i)$ represents the original output and $\Delta f$ is the sensitivity of $f$ given by $\Delta f = \max_{\mathcal{D}_{i}, \mathcal{D}_{i}^{\prime}}\left\|f\left(\mathcal{D}_{i}\right)-f\left(\mathcal{D}_{i}^{\prime}\right)\right\|_2$.
\end{theorem}
For detailed proof, please refer to the reference \cite{10.1561/0400000042}.
\par
We assume that the  Gaussian mechanism is adopted in our work because it is convenient to split DP noises obeying the Gaussian distribution into multiple shares \cite{abadi2016deep}. 

\subsection{DP-FedAVG}
\label{Sec:Preliminaries.dpfed}
FedAVG is the most commonly used model average algorithm in FL, and thereby FedAVG is used for our study. Based on previous works \cite{mcmahan2017communication,wu2020value,geyer2017differentially,wei2020federated}, we present the client based DP-FedAVG here to ease our following discussion. 

Without loss of generality, we assume that there are  $K$ clients. The client $k$ owns a private dataset $\mathcal{D}_k$ with cardinality $d_k$.  These clients target to train a model with parameters represented by the vector $\mathbf{w} \in \mathbb{R}^h$.  
In FedAVG, clients need to exchange model parameters for multiple  rounds with the PS. Each round is also called a global iteration. At the beginning of global round $t$, each participating client receives the global parameters $\mathbf{w}_{t}$ from the PS to conduct a number of local iterations. Then, clients return their locally updated model parameters plus DP noises to the PS. By receiving the computation results from a certain number of clients, the PS aggregates received parameters and embarks a new round of global iteration. The detail  of the DP-FedAVG algorithm is presented in {Algorithm \ref{FedAVG}}. 

\begin{algorithm}
    \caption{DP-FedAVG Algorithm}
    \label{FedAVG}
    \LinesNumbered
    \textbf{PS executes:}
    \\
    Initialize $\mathbf{w_0}$;
    \\
    \For{each round $t = 1,2,...$ }
    {
    $m \leftarrow$max$(C \times K,1)$
    \\
    $\mathcal{M}_t \leftarrow $(Random set of $m$ clients)
    \\
    \For{each client $k \in \mathcal{M}_t$ in parallel }
    {
    $\widetilde{\mathbf{w}}_{t+1}^k \leftarrow $ ClientUpdate($k,\mathbf{w}_t$, $p_k$)
    \par
    }
    $\mathbf{w}_{t+1} \leftarrow \sum_{k\in \mathcal{M}_t} \widetilde{\mathbf{w}}_{t+1}^k$
    }
    \ 
    \\
    \textbf{ClientUpdate$(k,\mathbf{w}_t, p_k)$} 
    \\
    $\mathcal{B} \leftarrow $(split $\mathcal{D}_k$ into batches of size $B$)
    \\
    \For{each local round $i$ from $1$ to $E$}
    {
    \For{batch $b \in \mathcal{B}$}
    {
    $\mathbf{w} \leftarrow \mathbf{w}-\eta \nabla \ell(w ; b)$
    }
    }
    $\sigma_k^2 \leftarrow$ (Gaussian Mechanism)
    \\
    $ \mathbf{n}_{t+1}^k \leftarrow \mathcal{N}(0, \sigma_k^2\mathbb{I}_h)$ 
    \\
    return  $p_k\mathbf{w} + \mathbf{n}_{t+1}^k$
\end{algorithm}

In {Algorithm \ref{FedAVG}},  $C$ is the fraction of clients that participate each global iteration, $\mathbf{n}_{t+1}^k$ is the Gaussian noise and  $p_k$ is the aggregation weight of client $k$, $\mathcal{M}_t$ is the set of clients that participate in round $t$. Usually, $p_k = d_k/(\sum_{i \in \mathcal{M}_t} d_i)$. $\mathcal{B}$ is the set of local sample batches, $E$ is the number of local iterations to be conducted and $\eta$ is the learning rate.

Let $f_k(\mathcal{D}_k, \mathbf{w}, p_k)$ represent the function returning the locally updated parameters with input $\mathbf{w}$ and $p_k$. The sensitivity of $f_k$ is denoted by $\Delta f_k$. We assume that the privacy budget of client $k$ is represented by $\epsilon_k$ and $\delta_k$. 
\begin{corollary}
Algorithm \ref{FedAVG} satisfies  $(\epsilon, \delta)$-differentially private, if $\mathbf{n}_{t+1}^k$ is sampled from $\mathcal{N}(0, \sigma_k^2\mathbb{I}_h)$ where $\sigma_k \geq \frac{c \Delta f_k}{\epsilon_k}$, $c^2 > 2 ln(1.25/\delta_k)$ and $\mathbb{I}_h$ is the $h\times h$ identity matrix.
\end{corollary}
Here $h$ is the model dimension. The proof is straightforward from Theorem~\ref{The:GaussM}.

According to Algorithm \ref{FedAVG}, the disturbance of the DP noises on the aggregated parameters is 
\begin{equation}
\label{EQ:Aggre}
  \mathbf{w}_{t+1} \leftarrow \sum_{k\in \mathcal{M}_t} p_k \mathbf{w}_{t+1}^k + \sum_{k\in \mathcal{M}_t}  \mathbf{n}_{t+1}^k.
\end{equation}
From the right hand side of Eq.\eqref{EQ:Aggre}, we can see that the first term represents the aggregated parameters while the second term represents the disturbance of the DP noises. They are independently generated by all participating clients, and therefore the variance of $\sum_{k\in \mathcal{M}_t}  \mathbf{n}_{t+1}^k$ is $\sum_{k\in \mathcal{M}_t}  \sigma_k^2\mathbb{I}_h$. 
Apparently, if the privacy budget $\epsilon_k$ is smaller, $\sigma_k$ is higher and  the total variance on the server side is higher. Our approach is to make these noises negatively correlated so that the aggregated noise variance can be reduced.



\section{NISS Algorithm}
\label{Sec:NISS}
In this section, we introduce the NISS algorithm in details and leave the analysis of the reduced variance on the aggregation and the  security analysis of NISS in the next section.

\subsection{Illustrative Example}
\label{Sec:NISS.ill}
Before diving into the detailed NISS algorithm, we present a concrete example to illustrate how NISS works.
According to Algorithm \ref{FedAVG}, $\mathbf{n}_k$ is sampled from $\mathcal{N}(0, \sigma^2\mathbb{I}_h)$ since the dimension of $\mathbf{w}$ is $h$. It means that  noises of $h$ dimensions are generated independently and the noise offset is conducted for each dimension independently.  Thus, 
to simplify our discussion, we only need to consider the noise $n^k$ for a particular dimension of  client $k$, and  $n^k$ is sampled from $\mathcal{N}(0, \sigma^2)$. 

According to the property of the Gaussian distribution, $n^k$ can be split into $v$ shares and each share is sampled from $\mathcal{N}(0, \frac{\sigma^2}{v})$. The client can send out $v$ negated share to $v$ neighboring clients.  If all clients conduct the same operation, the client is expected to receive $v$ noise shares from other clients as well, which can be denoted as $r^k =(r_1^k, \dots, r_v^k)$.  To ease our understanding, the process is illustrated in Fig.\ref{Fig:shares}.

\begin{figure}[ht]

\centering  
\subfigure[A client sends out negated noise shares. ]{
\label{fig:clientshare.a}
\includegraphics[width=0.3\textwidth]{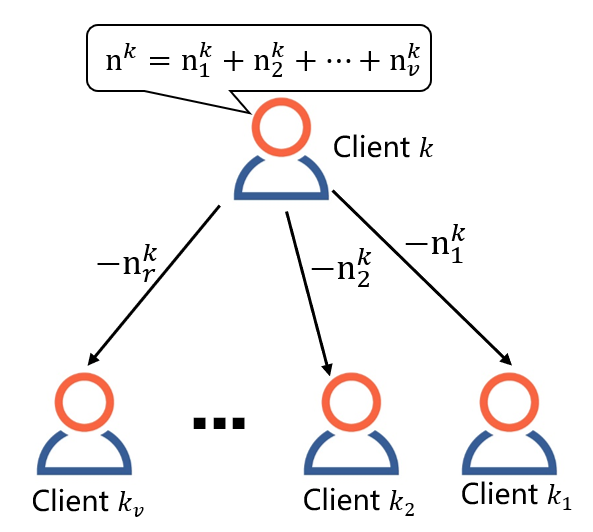}}
\subfigure[A client receives noise shares from other clients. ]{
\label{fig:clientshare.b}
\includegraphics[width=0.3\textwidth]{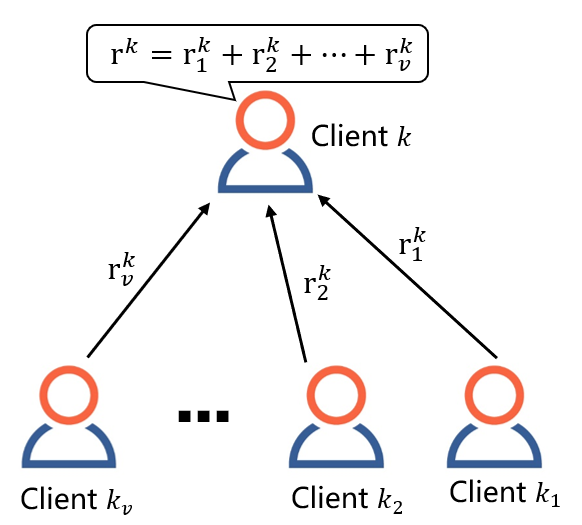}}
\caption{The workflow of NISS for a particular client. }
\label{Fig:shares}
\end{figure}

Then, the client adds both its own  noise $n^k$ and the sum of negated noise shares received from other clients to its parameter $w^k$ before  it submits $p_kw^k+ n^k+r^k$ to the PS. The $n^k+r^k$ can preserve the privacy of client $k$ while $r^k$ can be used to offset the noises generated by other clients by the PS. Since $v$ negated shares are generated randomly, no other client can exactly obtain the noise information of client $k$. In addition, the parameter information is only disclosed to the PS. As long as all other clients are trustworthy, these DP noises can be offset perfectly by negated noise shares.

\begin{algorithm}
    \caption{NISS Algorithm}
    \label{NISS}
    \LinesNumbered
    \textbf{PS executes:}
    \\
    Initialize $\mathbf{w_0},\sigma^2$;
    \\
    \For{each round $t = 1,2,...$ }
    {
    $m \leftarrow$max$(C \times K,1)$
    \\
    $\mathcal{M}_t \leftarrow $(Random set of $m$ clients)
    \\
    \For{each client $k \in \mathcal{M}_t$ in parallel }
    {
    $\widetilde{\mathbf{w}}_{t+1}^k \leftarrow $ClientUpdate($k,\mathbf{w}_t,p_k,\sigma^2$)
    \par
    }
    $\mathbf{w}_{t+1} \leftarrow \sum_{k\in \mathcal{M}_t} \widetilde{\mathbf{w}}_{t+1}^k$
    }
    \ 
    \\
    \ 
    \\
    \textbf{ClientUpdate$(k,\mathbf{w}_t,p_k,\sigma^2)$:} 
    \\
    $\mathcal{B} \leftarrow $(split $\mathcal{D}_k$ into batches of size $B$)
    \\
    \For{each local round $i$ from $1$ to $E$}
    {
    \For{batch $b \in \mathcal{B}$}
    {
    $\mathbf{w} \leftarrow \mathbf{w}-\eta \nabla \ell(w ; b)$
    }
    }
    $\widetilde{\mathbf{n}}_{t+1}^k \leftarrow$ ClientShare$(k,\sigma^2)$
    \\
    return $p_k\mathbf{w}+\widetilde{\mathbf{n}}_{t+1}^k$.
    \ 
    \\
    \ 
    \\
    \textbf{ClientShare$(k,\sigma^2)$}:
    \\
    $\sigma_k^2 \leftarrow $ (Gaussian mechanism)
    \\
    $\tau_k^2 \leftarrow$(Client $k$'s setting)
    \\
    $v \leftarrow \sigma_k^2/\sigma^2$
    \\
    $\mathcal{U} \leftarrow $(Random set of $v$ clients in this round)
    \\
    \For{$i = 1,2,....,v$}
    {
    $\mathbf{n}_i \leftarrow \mathcal{N}(0, \sigma_k^2\mathbb{I}_h)$
    \\
    $u_i \leftarrow$(Connect with $i$-th client in $\mathcal{U}$)
    \\
    Send $(-\mathbf{n}_i)$ to $u_i$
    \\
    Receive $\mathbf{r}_i$ from $u_i$
    \\
    $s_i \leftarrow \mathcal{N}(1,\tau_k^2)$
    }
    $\widetilde{\mathbf{n}} = \sum_{i=1}^v(\mathbf{n}_i+s_i\mathbf{r}_i)$
    \\
    return $\widetilde{\mathbf{n}}$
\end{algorithm}

\subsection{Algorithm Design}
We proceed to design the NISS algorithm based on the FedAVG algorithm introduced in the last section.

First of all, a tracker server is needed so that clients can  send and receive negated noise shares with each other. Each client needs to contact the tracker server to fetch a list of neighbor clients before it sends out negated noise shares. The tracker server is only responsible for recording live clients in the system and returning a random list of clients as neighbors for a particular client $k$. Obviously, the tracker server does not receive any noise information, and hence will not intrude user privacy. It can be implemented with light communication cost, similar to the deployment of the tracker server in peer-to-peer file sharing systems \cite{saroiu2001measurement}. 

In NISS, the operation of the PS is the same as that in FedAVG. The only difference lies in the operation of each client. Based on its own privacy budget and function sensitivity, the client $k$ needs to determine $\sigma_k$ so that $w^k+\mathcal{N}(0, \sigma_k^2)$ satisfies $(\epsilon_k, \delta_k) $-differentially privacy. Then, the client can determine the number of noise shares according to $v = \frac{\sigma_k^2}{\sigma^2}$ so that the client can generate $v$ noise shares and negated noise shares. Here $\sigma$ is a number much smaller than $\sigma_k$ and $\sigma$ can be a common value used by all clients. $\mathcal{N}(0, \sigma^2)$ is also called a unit noise. 

Because clients disclose their noise information with other clients, the gradient information can be  cracked to certain extent if some clients are not trustworthy and collude with the PS to intrude the privacy of a particular client. To prevent the leakage of user privacy, we propose to multiply a noise component $s_i$ to the received negated noise share $r_i$. $s_i$ is also sampled from the  Gaussian distribution $\mathcal{N}(1, \tau_k^2)$. Due to the disturbance of $s_i$, no other client and the PS can exactly crack the gradient information of the client. $\tau_k^2$ can be set according to the probability that other clients will collude with the PS. How to exactly set $\tau_k^2$ and the role of $s_i$ will be further analyzed in the next section.

By wrapping up, the details of the NISS algorithm is presented in {Algorithm \ref{NISS}}.

\section{Theoretic Analysis}
\label{Sec:Analysis}

In this section, we conduct analysis to show how much noise variance can be reduced by NISS on the PS side and  how the NISS algorithm defends against attacks. Based on our analysis, we also discuss the application of NISS in real FL systems.


\subsection{Analysis of Noise Offsetting}

Similar to Sec.\ref{Sec:NISS.ill}, to simplify our discussion, we only consider the noise offsetting for a particular dimension.  Let $n^k$ and $r^k$ denote noise shares and negated noise shares received from other clients for client $k$ respectively. 
Let $l_{k,i}$ denote the client that receive the $i$-th negated noise share from client $k$.

Based on {Algorithm \ref{NISS}}, the client $k$ uploads $p_kw^k+\sum_{i=1}^v (n^k_i+ r^k_i)$.  The aggregation conducted on the PS becomes
\begin{eqnarray}
 & &\sum_{k \in \mathcal{M}} p_kw^k+\sum_{k \in \mathcal{M}}\sum_{i=1}^v  (n^k_i+ s^k_i r^k_i), \nonumber\\
 & = &\sum_{k \in \mathcal{M}} p_kw^k+\sum_{k \in \mathcal{M}}\sum_{i=1}^v  (n^k_i- s^{l_{k,i}} n^k_i).\nonumber
\end{eqnarray}
Here $n_i^k$ is sampled from $\mathcal{N}(0, \sigma^2)$ and $s^{l_{k,i}}$ is sampled from $\mathcal{N}(1, \tau_l^2)$. $l$ is the abbreviation of $l_{k,i}$ if its meaning is clear from the context. Let $\mathbb{V}= \sum_{k \in \mathcal{M}}\sum_{i=1}^v  (n^k_i- s^{l_{k,i}} n^k_i)$ denote the aggregated DP noises and our study focuses on the minimization of $\mathbb{V}$.

Let us first analyze the variance of a particular noise share after offsetting. 
\begin{lemma}
\label{lemma:nsnr}
The variance of a noise share plus its negated share is:
\begin{equation}
     \mathbf{Var}[n_i^k-s^{l_{k,i}}n_i^k]  = \tau_l^2 \sigma^2.
\end{equation}
Here client $l$ receives the negated share of $n_i^k$ and $s^{l_{k,i}}$ is the noise imposed by client $l$.
\end{lemma}
\begin{proof}
According to the definition of the variance, we can obtain:
\begin{equation}
\begin{split}
     \mathbf{Var}[n_i^k-s^{l_{k,i}}n_i^k] & = \mathbf{Var}[(1-s^{l_{k,i}})n_i^k] 
    \\ 
    & = \mathbf{E}[(1-s^{l_{k,i}})^2(n_i^k)^2] -  (\mathbf{E}[(1-s^{l_{k,i}})(n_i^k)])^2
    \\
    & = \mathbf{E}[(1-s^{l_{k,i}})^2]\mathbf{E}[(n_i^k)^2] 
    \\
    & =  \mathbf{E}[(1-s^{l_{k,i}})^2]\sigma^2
    \\
    & = \mathbf{E}[1-2s^{l_{k,i}}+(s^{l_{k,i}})^2]\sigma^2
    \\
    & = (1-2\mathbf{E}[s^{l_{k,i}}]+\mathbf{E}[(s^{l_{k,i}})^2])\sigma^2
    \\
    & = (1-2+(1+\tau_l^2))\sigma^2
    \\
    & = \tau_l^2\sigma^2.
\end{split}
\end{equation}
The above formula holds because $n_i^k$ is sampled from $\mathcal{N}(0,\sigma^2)$ and $s^{l_{k,i}}$ is sampled from $\mathcal{N}(1,\tau_l^2)$. So we can obtain $\mathbf{E}[(n_i^k)^2] = \mathbf{Var}[n_i^k]+\mathbf{E}[n_i^k] = \sigma^2$ and $\mathbf{E}[(s^{l_{k,i}})^2] = 1 + \tau_l^2$ similarly.
Apparently, $n_i^k$ and $s^{l_{k,i}}$ are dependent according to our algorithm. 
\end{proof}

\begin{theorem}
\label{Th:PSnoise}
After noise offsetting, the variance of the aggregated noise  on the PS side is:
\begin{equation}
    \mathbb{V} = \sum_{k \in \mathcal{M}}\sigma_k^2 \tau_k^2.
\end{equation}
\begin{proof}
According to {Lemma \ref{lemma:nsnr}}, because each $n_i$ and $s_i$ are dependent, we can obtain:
\begin{equation}
    \begin{split}
     \mathbf{Var}[\sum_{k \in \mathcal{M}} \widetilde{n}^k] = \mathbf{Var}[\sum_{k \in \mathcal{M}} \sum_{l \in \mathcal{U}^k} (n_i^k-s^{l_{k,i}}n_i^k)].
    \end{split}
\end{equation}
Since client $k$ will send and receive $v_k$ noise shares and negated noise shares, thus each $\tau_l^2$ will be added $v_l$ times according to {Algorithm \ref{NISS}}. By substituting $l$ by $k$, we can obtain:
\begin{equation}
\begin{split}
    \mathbf{Var}[\sum_{k \in \mathcal{M}} \sum_{l \in \mathcal{U}^k} (n_i^k-s^{l_{k,i}}n_i^k)] = 
    \sum_{k \in \mathcal{M}} \sum_{l \in \mathcal{U}^k} \tau_l^2\sigma^2
    \\
    = \sigma^2 \sum_{k \in \mathcal{M}} v_k \tau_k^2
     =  \sum_{k \in \mathcal{M}} \sigma_k^2 \tau_k^2
\end{split}
\end{equation}
The third equality holds because $v_k = \sigma_k^2 / \sigma^2$.
\end{proof}
\end{theorem}

{\bf Remark:} From Theorem~\ref{Th:PSnoise}, we can observe that $\mathbb{V}= 0$ if $\tau = 0$ implying that DP noises are perfectly offset on the PS side. However, if $\tau = 1$, the value of $\mathbb{V}$  is the same as that without any noise offsetting. The value of $\tau$ depends on the trustworthy between clients. We will further discuss how to set $\tau$ after the security analysis in the next subsection.  

\subsection{ Security Analysis}

We conduct the security analysis through analyzing the privacy preservation for a particular client. We suppose that the target of a particular client is to satisfy the $(\epsilon_k, \delta_k)$-differentially private.

It is easy to understand that the NISS algorithm satisfies the $(\epsilon_k, \delta_k)$-differentially private by setting $\tau_k=0$ or $s^l=1$, if the PS and clients do not collude. What client $k$ submits to the PS is $w^k+\sum_{i=1}^v(n^k_i+r^k_i)$. The noise $\sum_{i=1}^v(n^k_i+r^k_i)$ is also a Gaussian random variable with variance $2\sigma_k^2$, and hence the NISS algorithm on client $k$ satisfies $(\epsilon_k, \delta_k)$-differentially private. Meanwhile, no other client can crack the parameter information since the parameter information is only disclosed to the PS.

However, it is not guaranteed that the PS never colludes with clients. To conduct more general analysis, we assume that there is $\rho\in [0,1]$ fraction of other clients will collude with the PS. The problem is how to set $\tau_k$ so that the NISS algorithm on client $k$ can still satisfies 
$(\epsilon_k, \delta_k)$-differentially private.

Let $\mathcal{U}^k$ represent the set of clients that client $k$ will contact. There is no prior knowledge about which client will collude with the PS. The tracker server randomly select clients for $\mathcal{U}^k$. It implies that $\rho$ fraction of $\mathcal{U}^k$ will disclose the noise share information with the PS. We use $\mathcal{U}_{\rho}^k$ to denote the clients who collude with the PS and $\mathcal{U}_{1-\rho}^k$ to denote the clients who do not collude. Apparently, the size of $\mathcal{U}_{\rho}^k$ and $\mathcal{U}_{1-\rho}^k$ are $\rho v$ and $(1-\rho)v$. Thus, the effective noise uploaded by client $k$ becomes $\sum_{i \in \mathcal{U}_{1-\rho}^k }(n^k_i+s^{l_{k,i}} r^k_i) + \sum_{i \in \mathcal{U}_{\rho}^k} s^{l_{k,i}} r^k_i$. To ensure that  $(\epsilon_k, \delta_k)$-differentially private can be satisfied, it requires $\mathbf{Var}[\sum_{i \in \mathcal{U}_{1-\rho}^k }(n^k_i+s^{l_{k,i}} r^k_i) + \sum_{i \in \mathcal{U}_{\rho}^k} s^{l_{k,i}} r^k_i]\geq \sigma_k^2$. It turns out that
\begin{theorem}
\label{THE:tau}
If $n^k$ sampled from $\mathcal{N}(0, \sigma_k^2)$ can make $w^k+n^k$ satisfy $(\epsilon_k, \delta_k)$-differentially private, the NISS algorithm satisfies $(\epsilon_k, \delta_k)$-differentially private as long as $\tau_k^2\geq \max\{2\rho-1, 0\}$. Here $\rho\in[0,1)$ represents the percentage of other clients that collude with the PS. 
\end{theorem}
The detailed proof is presented in Appendix A.

{\bf Remark:} It is worth to mention the special case with $\rho=1$. According to Theorem~\ref{THE:tau}, $\tau=1$ and $s^{l_{k,i}}$ is sampled from $\mathcal{N}(1,1)$ if $\rho =1$. In this case, $\mathbf{Var}[\sum_{i \in \mathcal{U}_{1-\rho}^k }(n^k_i+s^{l_{k,i}} r^k_i) + \sum_{i \in \mathcal{U}_{\rho}^k} s^{l_{k,i}} r^k_i]=\mathbf{E} [(\sum_{i=1}^v r^k_i)^2]$. According to the central limit theorem,  as long as $v\gg 1$, we have $\mathbf{E} [(\sum_{i=1}^v r^k_i)^2] \approx \sigma_k^2$. Thus, if $\rho=1$ and $\tau =1$, it implies that $\mathbf{Var}[\sum_{i \in \mathcal{U}_{1-\rho}^k }(n^k_i+s^{l_{k,i}} r^k_i) + \sum_{i \in \mathcal{U}_{\rho}^k} s^{l_{k,i}} r^k_i]=\sigma_k^2$ and $\mathbb{V}= \sum_k \sigma_k^2$. The variance of the aggregated noise on the PS is the same as that without any offsetting operation. In this extreme case, there exists a trade-off between model accuracy and privacy protection. One cannot improve the model accuracy without compromising privacy protection.

\subsection{Application of NISS in Practice}

As we have discussed in the last section, $\rho$ is a vital parameter. Our analysis uses $\rho$ to cover all cases with different fractions of malicious clients colluding with the PS.   If $\rho$ is close to $1$, it will significantly impair the performance of the NISS algorithm. 
In practice, $\rho$ can be set as a small value, which can be illustrated from two perspective. 

Firstly,  most FL systems are of a large-scale with tens of thousands of clients. If there are more normal clients, the fraction of malicious clients that will collude with the PS will be a smaller value. Secondly, our analysis is based on the assumption that $\mathcal{U}^k$ is randomly selected by the tracker server. In fact, clients can play coalitional  game with other clients they trust. For instance, the IoT devices of the same system can trust each other substantially. They can share negated noise information with each other by setting a small $\tau_k^2$  since the probability that neighboring clients collude with the PS is very low. From this example, we can also conclude that the NISS algorithm is particularly applicable for FL across multiple IoT systems. IoT devices in the sample system can form a coalition so that the variance of the aggregated noise is minimized. Besides,  devices within the same system can be connected with high speed networks so that the communication overhead to transmit noise shares is insignificant.

\section{Experiment}

\begin{figure*}[htbp]
\includegraphics[width=\textwidth]{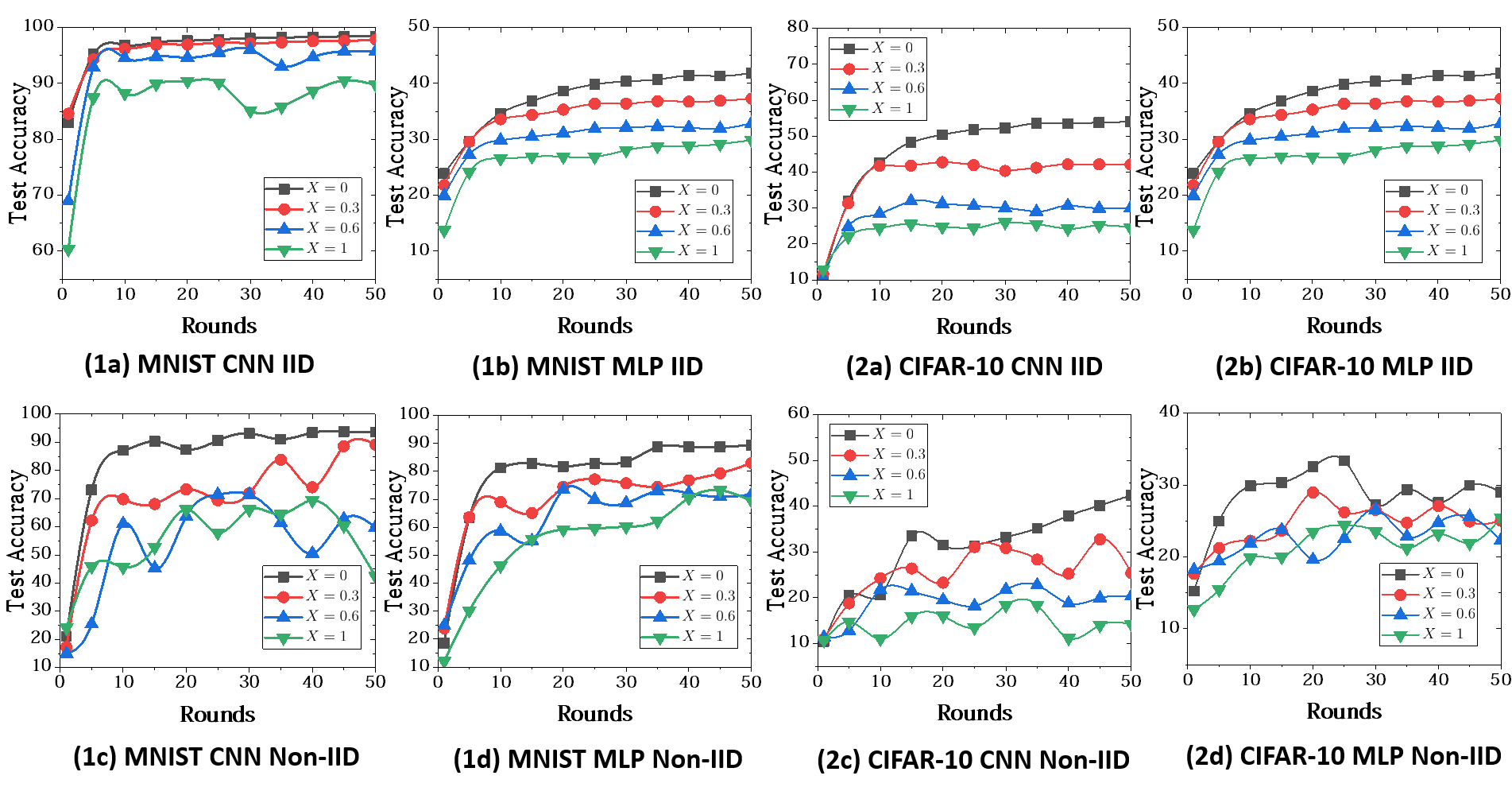}
\caption{$X = \tau_k^2$. If $X=0$,   FedAVG. If $X=1$,  DP-FedAVG. For the rest, NISS with $X =0.3$ and $X = 0.6$.}
\label{experiments}
\end{figure*}

In this section, we conduct experiments with MNIST and CIFAR-10 to evaluate the performance of NISS. 

\label{Sec:Experiment}
\subsection{Experimental Setup}

\subsubsection{Simulation Settings}
Based on \cite{abadi2016deep}, we use the Gaussian mechanism to add noises to local model parameters. We use the same experimental settings as in \cite{mcmahan2017communication}, \cite{wei2020federated} and \cite{abadi2016deep}. The FL settings of our experiment are as follows: 
\begin{itemize}
    \item Number of users: $K = 100$
    \item User fraction rate: $C = 0.3$
    \item Local minibatch size: $B =10$
    \item Learning rate: $\eta = 0.01$
    \item Number of local round: $E = 5$
    \item Unit noise variance: $\sigma^2 = 0.01$
    \item DP parameters: $\epsilon = 10, \delta = 0.0001$
\end{itemize}
\par
In addition, to achieve DP-FedAVG, we use the norm clipping technique with a clipping threshold $\zeta$ to restrict the range of client's gradients. If a client’s some gradient exceeds $\zeta$, it will be clipped to $\zeta$. The details of its mechanism can be found in \cite{wei2020federated}. For our experiments, we set $\zeta = 3$.  

\par
We use Pytorch \cite{paszke2017automatic} as our experimental environment. The experiments run on the computer with a processor with six 2.6GHz Intel i7 cores. The computer is equipped with 16GB RAM and a GPU of AMD Radeon Pro 5300M.

\subsubsection{Training Models and Datasets}
To make our experiment more comprehensive, we set up three different scenarios.
Firstly, we use public dataset MNIST and CIFAR-10 as our experimental data set. The MNIST dataset of handwritten digits contains 60,000 $28\times28$ grayscale images of the 10 digits with 50,000 training images and 10,000 test images. The CIFAR-10 dataset also consists of 60,000 $32 \times 32$ colour images in 10 classes, with 6,000 images per class. There are 50,000 training images and 10,000 test images.
Secondly, We use different neural network structures which are similar to those in \cite{mcmahan2017communication} \cite{abadi2016deep} and \cite{li2020privacy2} . i) A Convolutional Neural Network (CNN) with two $5 \times 5$ convolution layers, a fully connected layer with 512 units
and ReLU activation, and a final softmax output layer. ii) A Multilayer Perceptron (MLP) with 2-hidden layers with 200 units each using ReLU activations.
Thirdly, we split the dataset in IID and non-IID settings. For the IID setting, the data is shuffled and partitioned into 100 users each receiving the same size of examples. For non-IID setting, the data is sorted by labels and divided into different partitions. Then we distribute them to each client so that each client will receive a non-IID dataset.


\subsubsection{Metrics and Baselines}

We use the model accuracy on the test dataset to evaluate the accuracy performance of the NISS algorithm. Meanwhile, we implement FedAVG and DP-FedAVG algorithms as baselines in our experiments. 

In addition,  we also use the method in \cite{zhu2019deep} to detect the effect of NISS on privacy protection. we can evaluate leak-defence of a model average algorithm by determining whether the effective information can be recovered from one picture of CIFAR-100 or not. Similar to \cite{zhu2019deep}, we adopt the  DLG loss as the metrics. The method uses randomly initialized weights and uses L-BFGS \cite{liu1989limited} to match gradients from all trainable parameters. The DLG loss is the gradient match loss for L-BFGS. The lower the DLG loss is, the more information leaks, then the final recovered image will be clearer.

\subsection{Experiment results}

\subsubsection{Model Accuracy}

Fig.\ref{experiments} shows the results on the test accuracy of training models. Since we set up three different scenarios: different dataset, IID or non-IID and different neural network, we conducted eight sets of experiments. Here for feasibility and clarity, we uses $X$ to denote $\tau_k^2$ which is the variance of $s_i$. Then $X = 0$ means perfect offsetting by NISS which is equal to the effect of FedAVG. $X = 1$ means the variance of the aggregated noise on the PS side is $\sum_{k \in \mathbf{M}} \sigma_k^2$ which is the same as DP-FedAVG. Thus we use $X = 0$ and $X = 1$ to denote FedAVG and DP-FedAVG. From Fig.\ref{experiments}, we can see, by tuning $X$, the test accuracy of training model is increasing which means all clients are adding more noise and cause the variance of the aggregated noise on the PS side to increase. The higher $X$ is, the larger the variance of added noises is, and the more significant the accuracy
 deteriorates. This is consistent with our analysis. When the client data is IID, our NISS algorithm can increase the test accuracy by about $12 \sim 13 \%$ on MNIST and $15 \sim 25 \%$ on CIFAR-10 if all clients will not collude with the PS, namely, perfectly offsetting. This is because CIFAR-10 are all three-channel color picture, and the amount of noise has a higher impact on the accuracy. When the client data is non-IID, the test accuracy on MNIST increases by $30 \sim 40 \%$ and for CIFAR-10 the test accuracy is $10 \sim 20 \%$ higher. In addition, note that the test accuracy of CIFAR-10 is low because MLP model is too simple for training CIFAR-10 and non-IID data can cause it a low testing accuracy, this can be found in \cite{zhao2018federated}. Fig.\ref{experiments} also shows the trade-off between model accuracy and privacy protection. If we increase $X$, the accuracy will decrease and if we decrease $X$, the accuracy will increase. 
 \par
 In summary, when $X = 0$, since the noise can be offset perfectly, the model accuracy given by NISS is very close to that of FedAVG on the whole and better than that of the DP-FedAVG algorithm if all the clients will not collude with the PS. And even if some clients collude with the PS, by tuning $X$, each client can protect its privacy but the model accuaracy will decrease.

\begin{figure}[htbp]
\includegraphics[width=0.5\textwidth]{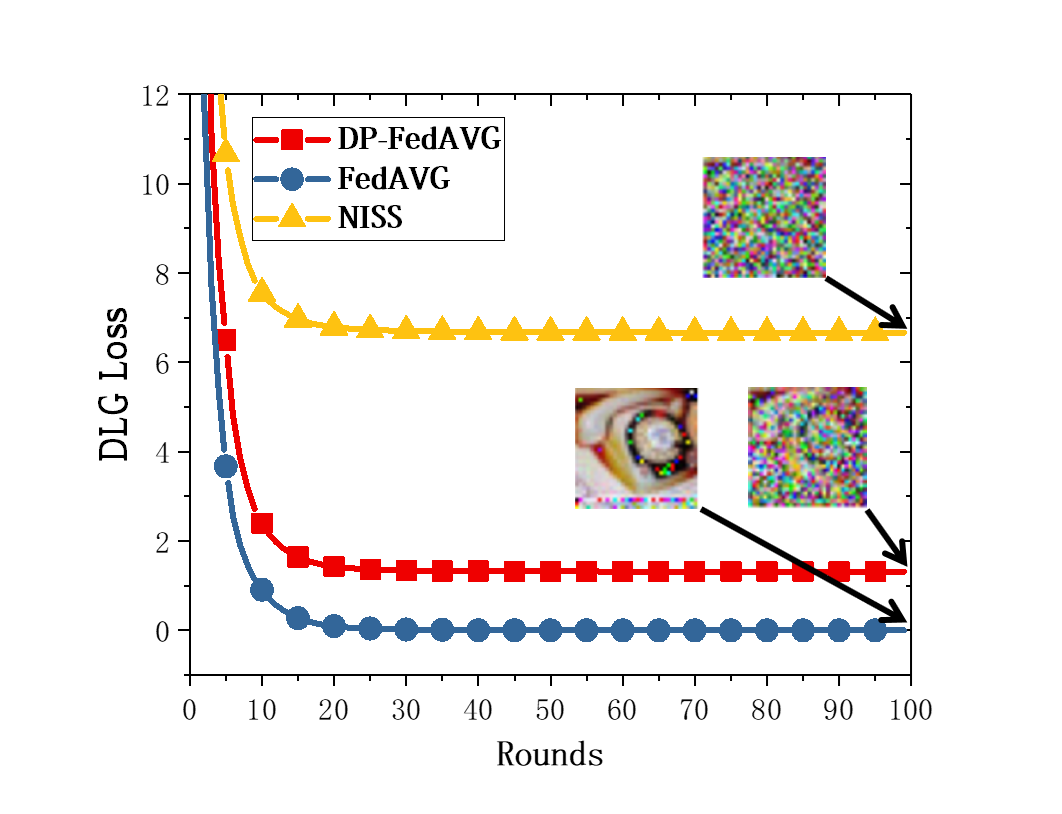}
\caption{DLG loss for FedAVG, DP-FedAVG and NISS. The above images show the effect of finally recovering image by DLG in \cite{zhu2019deep} and the original image is a "telephone".}
\label{experiments:dlg}
\end{figure}



\subsubsection{Privacy Protection}
In order to test the degree of privacy protection of the client for gradient information, we use method in \cite{zhu2019deep} to test the leak-defence of FedAVG, DP-FedAVG and NISS. We use the gradients from FedAVG, DP-FedAVG and NISS to run DLG. Fig.\ref{experiments:dlg} shows the results of DLG loss and the image it finally recovered. The lower the DLG loss is, the more information leaked, then the final recovered image will be clearer. From Fig.\ref{experiments:dlg}, we observe that our NISS algorithm almost does not leak any sensitive information, while DP-FedAVG may leak partial information about the privacy and FedAVG can not prevent the leakage of sensitive information totally.

\par
In summary, the above experiments demonstrate that our NISS algorithm can achieve extraodinary performance. When clients do not collude with the PS, our NISS can achieve the same accuracy as that of FedAVG which is better than DP-FedAVG and better privacy protection due to its large scale of noise for a single client. If some clients collude with the PS, the client can set its $\tau_k^2$ to protect its privacy. Our experiments also show the trade-off between model accuracy and privacy protection. If clients set a higher $\tau_k^2$, the accuracy will be lower and vice versa.

\section{Conclusion}
\label{Sec:Conclusion}
In this work, we propose a novel algorithm called NISS to offset the DP noises independently generated by clients in FL systems. NISS is a method for clients to generate negatively correlated noises.  Intuitively, each client splits its noise into multiple shares. Each share is negated and sent out to a neighboring client. Each client uploads its parameter plus its own noise and all negated noise shares received from other neighbors. A noise share of a particular client can be potentially offset by its negated value uploaded by another client.  We theoretically prove that the NISS algorithm can effectively reduce the variance of the aggregated noise on the PS so as to improve the model accuracy in FL. Experiments with MNIST and CIFAR-10 datasets are carried out to verify our analysis and demonstrate the extraordinary performance achieved by NISS.

\appendices
\section{Proof of theorem \ref{THE:tau}}
\begin{proof}
We will calculate $\mathbf{Var}[\sum_{i \in \mathcal{U}_{1-\rho}^k }(n^k_i+s^{l_{k,i}} r^k_i) + \sum_{i \in \mathcal{U}_{\rho}^k} s^{l_{k,i}} r^k_i]$ first. For $i \in \mathcal{U}_{1-\rho}^k$ and $i \in \mathcal{U}_{\rho}^k$, we will discuss separately. Firstly, for $i \in \mathcal{U}_{1-\rho}^k$, we have:
\begin{equation}
\begin{split}
     \mathbf{Var}[\sum_{i \in \mathcal{U}_{1-\rho}^k }(n^k_i+s^{l_{k,i}} r^k_i)] &  = 
      \sum_{i \in \mathcal{U}_{1-\rho}^k } \mathbf{Var}[(n^k_i+s^{l_{k,i}} r^k_i)] 
      \\ &
       = \sum_{i \in \mathcal{U}_{1-\rho}^k } \mathbf{Var}[(n^k_i]+ \mathbf{Var}[s^{l_{k,i}}r^k_i)] 
       \\ &
       = \sum_{i \in \mathcal{U}_{1-\rho}^k } \sigma^2 + \mathbf{E}[(s^{l_{k,i}}r^k_i))^2] 
       \\ &
       = \sum_{i \in \mathcal{U}_{1-\rho}^k } \sigma^2 + (\tau_k^2+1) \sigma^2
       \\ &
      = (1-\rho) (\tau_k^2+2) \sigma_k^2
\end{split}
\end{equation}
Secondly, for $i \in \mathcal{U}_{\rho}^k$, note that here $r^k_i$ is no longer a random variable and it is a certain number which we can approximate using central limit theorem, then we can calculate it as:
\begin{equation}
    \begin{split}
        \mathbf{Var}[\sum_{i \in \mathcal{U}_{\rho}^k} s^{l_{k,i}} r^k_i] & = 
        \sum_{i \in \mathcal{U}_{\rho}^k}\mathbf{Var} [ s^{l_{k,i}} r^k_i] 
        \\ &
         = \sum_{i \in \mathcal{U}_{\rho}^k} (r^k_i)^2 \mathbf{Var} [ s^{l_{k,i}}] 
         \\ &
         = \tau_k^2 \sum_{i \in \mathcal{U}_{\rho}^k}  (r^k_i)^2 
    \end{split}
\end{equation}
According to the central limit theorem, as long as $\rho v \gg 1$,  $\sum_{i \in \mathcal{U}_{\rho}^k} (r^k_i)^2 \approx \mathbf{E}[\sum_{i \in \mathcal{U}_{\rho}^k} (r^k_i)^2] = \rho v \sigma^2 = \rho \sigma_k^2 $. Thus, we have:
\begin{equation}
     \mathbf{Var}[\sum_{i \in \mathcal{U}_{\rho}^k} s^{l_{k,i}} r^k_i] =  \tau_k^2  \rho \sigma_k^2
\end{equation}
Then, we can obtain:
\begin{equation}
\begin{split}
     \mathbf{Var} & [\sum_{i \in \mathcal{U}_{1-\rho}^k }(n^k_i+s^{l_{k,i}} r^k_i) + \sum_{i \in \mathcal{U}_{\rho}^k} s^{l_{k,i}} r^k_i] 
       \\  & 
       =  \mathbf{Var}[\sum_{i \in \mathcal{U}_{1-\rho}^k }(n^k_i+s^{l_{k,i}} r^k_i)] +  \mathbf{Var}[\sum_{i \in \mathcal{U}_{\rho}^k} s^{l_{k,i}} r^k_i] 
       \\ & 
       = (1-\rho) (\tau_k^2+2) \sigma_k^2+  \tau_k^2  \rho \sigma_k^2
\end{split}
\end{equation}
To ensure that $(\epsilon_k,\delta_k)$-differentially private, it requires that:
\begin{equation}
     (1-\rho) (\tau_k^2+2) \sigma_k^2+  \tau_k^2  \rho \sigma_k^2 \geq \sigma_k^2
\end{equation}
Then we have:
\begin{equation}
    \tau_k^2 \geq 2 \rho -1
\end{equation}
Hence we can obtain $\tau_k^2 \geq \max \{ 2\rho-1,0 \}$.
\end{proof}


\ifCLASSOPTIONcaptionsoff
  \newpage
\fi




\bibliography{ref}
\bibliographystyle{unsrt}
\end{document}